\documentclass{statsoc}
\usepackage{lmodern}
\usepackage[a4paper]{geometry}
\usepackage{graphicx}
\usepackage[textwidth=8em,textsize=small]{todonotes}
\usepackage{amsmath}
\usepackage{natbib}
 \usepackage[colorlinks]{hyperref}  

 \usepackage{booktabs}
 \usepackage{amssymb,bbm}
 \usepackage{arydshln}
 \usepackage{latexsym,bm}
 \usepackage{epstopdf}
 \usepackage{enumerate}
 \usepackage{multirow}
 \usepackage{dsfont}
\usepackage{pifont}
\newcommand{\cmark}{\ding{51}}%
\newcommand{\xmark}{\ding{55}}%

\newtheorem{theorem}{Theorem}
\newtheorem{assumption}{Assumption}
\newtheorem{lemma}{Lemma}
\newtheorem{corollary}{Corollary}

\newcommand{\Var}{\mathrm{Var}}

\newcommand{\Cov}{\mathrm{Cov}}

\newcommand{\E}{\mathrm{E}}
\newcommand{\AC}{\mathcal{C}}
\newcommand{\diag}{\operatorname{diag}}
\newcommand{\mbf}{\mathbf}
\newcommand{\tabincell}[2]{\begin{tabular}{@{}#1@{}}#2\end{tabular}}

\title[PCA FOR TIME SERIES]{Asymptotic Theory of Principal Component Analysis for\\ Time Series Data with Cautionary Comments}

\author[Zhang X. and Tong H.]{Xinyu Zhang$^1$ and Howell Tong$^2$} 
\address{$^1$Center for Statistical Science, Department of Industrial Engineering, Tsinghua University,
Beijing,
China.} 
\address{$^2$School of Mathematical Sciences, University of Electronic Science and Technology of China,
Chengdu,
China;
Tsinghua University,
Beijing,
China;
London School of Economics and Political Science,
London,
UK.}

\graphicspath{{./art/}} 
\begin{document}

\begin{abstract}
Principal component analysis (PCA) is a most frequently used statistical tool in almost all branches of data science. 
However, like many other statistical tools, there is sometimes the risk of misuse or even abuse. 
In this paper, we highlight possible pitfalls in using the theoretical results of PCA based on the assumption of independent data when the data are time series.
For the latter, we state with proof a central limit theorem of the eigenvalues and eigenvectors (loadings), give direct and bootstrap estimation of their asymptotic covariances, and assess their efficacy via simulation.
Specifically, we pay attention to the proportion of variation, which decides the number of principal components (PCs), and the loadings, which help interpret the meaning of PCs.
Our findings are that while the proportion of variation is quite robust to different dependence assumptions, the inference of PC loadings requires careful attention.
We initiate and conclude our investigation with an empirical example on portfolio management, in which the PC loadings play a prominent role.  It is given as a paradigm of correct usage of PCA for time series data.

\end{abstract}


\section{Introduction}\label{sec:intro}
Principal component analysis (PCA) is probably one of the most widely used statistical tools in statistics and fields well beyond. 
Comprehensive summaries and numerous empirical examples of PCA can be found in books such as \cite{flury1997}, \cite{jolliffe2002} and \cite{tsay2005,tsay2013}.
Now, classical results for PCA are obtained under the assumption of \emph{independent} (vectorial) data.
However, applications to dependent data abound in which dependence is ignored. 
To cite but a few examples, we mention \cite{stone1947} in economics, \cite{craddock1965} and \cite{maryon1979} in meteorology, \cite{ahamad1967} for crime rates, \cite{feeney1967}, \cite{tsay2005}, \cite{wei2018} and \cite{keijsers2019} in finance, and others. 
The issues that  arise naturally are the defects, in both theory and practice, of this misuse, and a correct inference of PCA for time series. These  are the issues that this paper aims to resolve.

For independent data, classical results about the limiting distribution of the principal components (PCs) can be found in \cite{anderson2003} and are widely known. 
For dependent data, a comprehensive and practically relevant central limit theorem for time series data is lacking in the literature. 
In this paper, we will state with proof a central limit theorem for stationary and ergodic multivariate linear time series by building on \cite{hannan1976asymptotic}.
Besides the limiting distributions of estimated eigenvalues and estimated eigenvectors, 
we also give the limiting distribution of the proportion of variation to facilitate the decision of PC numbers.
Then, we introduce two special cases or approximations of the theorem when the processes are Gaussian, which assume simpler dependence structure, one being the classical result for independent data.
The efficacy of our theorem and the two special cases are assessed via simulation for Gaussian time series.
It turns out that the proportion of variation is a statistic quite robust to different modes of dependence, in that they have only negligible effect.
However, the situation with the loadings is markedly different.
In the simulation, the efficacy of our theorem is verified, while the two special cases behave poorly for the loadings numerically and thus are not practically useful for time series data.

Inference based on the limiting distribution necessitates the estimation of the asymptotic covariances of eigenvalues and eigenvectors.
For this, we give two methods: the direct and the bootstrap estimation.
The former is recommended for Gaussian processes and some light-tail processes, while the latter is applicable to general processes.

In practice, PCA for time series is widely implemented but inference is often missing or inappropriate.
To illustrate, we focus on a specific application field of PCA, namely the stock portfolio management.
Portfolios can be constructed based on the principal components and they are named ``principal portfolios"  \citep{partovi2004principal}.
In \cite{pasini2017principal}, it is shown that these principal portfolios can get enhanced returns and financial risk control.
Further, they can be viewed as individual assets, based on which all the asset allocation strategies can be applied to reduce risk.
See, e.g., \cite{yang2015application}.
An important benefit is that in  principal portfolios 
investors do not need to concern themselves with the co-movements among  assets.
The first principal portfolio is normally understood as the market component with roughly equal contributions of the underlying stocks.
A number of following principal portfolios always represent synchronized fluctuations that only happen to a group of stocks.

The obtained principal portfolios are each a linear combination of all the stocks. 
However, although diversification can help reducing risks, the benefit will decrease as the number of stocks increases. 
Investors always prefer a portfolio that is sufficiently diverse, at the same time being as small as possible.
Thus, implementing inference to the principal portfolio loadings is important since a much smaller portfolio could be obtained by discarding stocks with insignificant loadings.
The sparser loadings will also lead to a more refined interpretation.
In the real data example, we indicate how possible  principal portfolios could be constructed by performing a proper inference on the loadings of the PCA.
By so doing, we give a paradigm of  correct procedures of implementing PCA on time series, and highlight the fact that without an appropriate inference, mis-interpretations in general, and wrong portfolio strategy in particular, can arise.

The paper is organised as follows. In Section \ref{sec:the}, we give with proof a central limit theorem of the principal components for time series, and obtain the limiting distributions of eigenvalues, eigenvectors and the proportion of variation. 
To compare, in Section \ref{sec:special} we introduce two special cases of the theorem for Gaussian processes, one being the classical result for independent data. 
To estimate the asymptotic covariances of the eigenvalues and eigenvectors, two methods, the direct and the bootstrap estimation, are given in Section \ref{sec:numerical}.
In Section \ref{sec:simu}, we conduct simulations to assess the efficacy of the above theoretical results and 
highlight pitfalls of using classical results (based on the assumption of independent data) for time series. 
We also evaluate the numerical performance of the two estimation methods.
An empirical example of the stocks portfolio is given in Section \ref{sec:eg} to show a paradigm of  correct  implemention of PCA for time series data. 
We conclude in Section \ref{sec:con}. 
Details about the simulation are in  Appendix \ref{sec:app}.

\section{A central limit theorem}\label{sec:the}

Let $\mathbf{X}(t)$ be a stationary and ergodic $p$-dimensional linear vector process generated as
\begin{equation}
\mathbf{X}(t)=\sum_{j=0}^{\infty} \mathbf{G}(j) \mathbf{e}(t-j), \quad t \in \{...,-1,0,1,...\},
\label{eq:xt}
\end{equation}
where $\mathbf{e}(t)$ is $p$-dimensional, and $\E(\mathbf{e}(t))=\mathbf{0}_p$, $\E(\mathbf{e}(t_1)\mathbf{e}(t_2)^T)=\mathbf{K}\mathds{1}(t_1=t_2)$ with $\mathbf{K}$ a $p \times p$ nonsingular matrix.
Here $\mathds{1}(\cdot)$ is the indicator function. 
The Kronecker delta $\delta_{t_1t_2}$ is 1 when $t_1=t_2$ and 0 otherwise. 
Here $\mathbf{G}(j)$'s are $p \times p$ matrices with 
$\sum_{j=0}^{\infty} \operatorname{tr} \{\mathbf{G}(j) \mathbf{K G}(j)^T\}<\infty$.

Let $\mathbf{\Gamma}(s)=\mathbf{\Gamma}_{\mbf{X}}(s)=\E(\mathbf{X}(t+s)\mathbf{X}(t)^T)$ for $s \in \{...,-1,0,1,...\}$.
Specifically, let $\mathbf{\Gamma}\equiv \mathbf{\Gamma}(0)$ be the covariance matrix.
Then, the spectral density matrix of $\mathbf{X}(t)$ at frequency $\omega$ is defined as
\begin{equation*}
    f(\omega)=(2\pi)^{-1}\sum_{s=-\infty}^{\infty}\exp\{-\mathrm{i}s\omega\}\mathbf{\Gamma}(s),
\end{equation*}
where $\mathrm{i}=(-1)^{1/2}$.
Let the transfer function be $h(\omega)=\sum_{j=0}^{\infty} \mathbf{G}(j)e^{-\mathrm{i}j\omega}$.

Suppose we have observations $\{\mathbf{X}(1),...,\mathbf{X}(n)\}$. Define the sample covariance matrix as $\mathbf{C}_n=(1/n) \sum_{t=1}^n (\mathbf{X}(t)-\bar{\mathbf{X}})(\mathbf{X}(t)-\bar{\mathbf{X}})^T$, where $\bar{\mathbf{X}}$ is the sample mean. 
Denote the $(i,j)$-th component of $\mathbf{C}_n$, $\mathbf{\Gamma}$, $h(\omega)$ and $f(\omega)$ by $\mathbf{C}_{n,ij}$, $\mathbf{\Gamma}_{ij}$, $h_{ij}(\omega)$ and $f_{ij}(\omega)$, respectively. Denote the conjugate of $f_{ij}(\omega)$ by $\overline{f_{ij}(\omega)}$.

Let $\bm{\lambda}=(\lambda_1,...,\lambda_p)$ and $\bm{\Phi}=(\bm{\phi}_1,..,\bm{\phi}_p)$
be the eigenvalues and eigenvectors of $\mathbf{\Gamma}$, respectively.
Assume all eigenvalues  are positive and distinct. Specifically, $\lambda_1>\lambda_2>...>\lambda_p>0$.
Let $\mathbf{l}=(l_1,...,l_p)$ and $\mathbf{A}=(\mathbf{a}_1,..,\mathbf{a}_p)$ be the eigenvalues and eigenvectors of $\mathbf{C}_n$, respectively.
Let $\bm{\Lambda}=\diag(\bm{\lambda})$ and $\mathbf{L}=\diag(\mathbf{l})$. Then we can express $\mathbf{\Gamma}=\mathbf{\Phi \Lambda \Phi}^T$ and $\mathbf{C}_n=\mathbf{ALA}^T$.

When applying PCA, we pay attention to
the number and loadings of PCs.
The number of PCs are usually decided by the proportion of variation
\begin{equation}\label{eq:pro}
    {\gamma}_k=\sum_{i=1}^k \lambda_i/\sum_{i=1}^p \lambda_i
    =\sum_{i=1}^k \lambda_i/\mathrm{tr}(\mbf{\Gamma}),
    \quad 
    {r}_k=\sum_{i=1}^k l_i/\sum_{i=1}^p l_i
    =\sum_{i=1}^k l_i/\mathrm{tr}(\mathbf{C}_n),
\end{equation}
for $k=1,...,p$.
Let $\bm{\gamma}=(\gamma_1,\ldots,\gamma_p)$ and $\mbf{r}=(r_1,\ldots,r_p)$.
To derive the asymptotics, we introduce two assumptions.
\begin{assumption}\label{ass1}
Let $\mathcal{F}_t$ be the sub $\sigma$-algebra generated by $\{\mathbf{X}(n), n \leq t\}$. 
Denote the $a$-th component of $\mathbf{e}(t)$ as $e_a(t),a=1,...,p$. Suppose that 
\begin{align}
\begin{split}
&\E\{e_a(t)|\mathcal{F}_{t-1}\},\quad
\E\{e_a(t)e_b(t)|\mathcal{F}_{t-1}\},\\
&\E\{e_a(t)e_b(t)e_c(t)|\mathcal{F}_{t-1}\},\quad
\E\{e_a(t)e_b(t)e_c(t)e_d(t)|\mathcal{F}_{t-1}\},
\label{eq:moments}
\end{split}
\end{align}
are constants for all $a,b,c,d=1,...,p$.
\end{assumption}

\begin{assumption}\label{ass2}
Suppose $f_{kk}(\omega), k=1,2,...,p,$ are all square integrable.
\end{assumption}

The following lemma is from \cite{hannan1976asymptotic}, after correcting some (possibly typographical) errors in the second term of his equation (3).

\begin{lemma}
Suppose $\mathbf{X}(t)$ is stationary, ergodic, and generated by (\ref{eq:xt}).
Then, under Assumptions \ref{ass1} and \ref{ass2}, for $i,j,k,l=1,...,p$,
$\sqrt{n}\{\mathbf{C}_{n,ij}-\mathbf{\Gamma}_{ij}\}$ has a joint asymptotic normal distribution with zero mean, and the asymptotic covariance between $\sqrt{n}\{\mathbf{C}_{n,ij}-\mathbf{\Gamma}_{ij}\}$ and $\sqrt{n}\{\mathbf{C}_{n,kl}-\mathbf{\Gamma}_{kl}\}$ is 
\begin{align}
\begin{split}
&2 \pi \int_{-\pi}^{\pi}\left\{f_{i k}(\omega) \overline{f_{j l}(\omega)}+f_{i l}(\omega) \overline{f_{j k}(\omega)}\right\} d \omega\\
+&(4 \pi^2)^{-1}\sum_{a,b,c,d=1}^p\kappa_{abcd}\int_{-\pi}^{\pi}\int_{-\pi}^{\pi}h_{ia}(\omega_1)\overline{h_{jb}(\omega_1)} \overline{h_{kc}(\omega_2)}h_{ld}(\omega_2)d\omega_1 d\omega_2,
\label{1}
\end{split}
\end{align}
where $\kappa_{abcd}$ is the joint fourth cumulant of $e_a(t)$, $ e_b(t)$, $ e_c(t)$ and $e_d(t)$. Specifically,
\begin{align*}
\kappa_{abcd}=&\E\{e_a(t) e_b(t) e_c(t) e_d(t)\}-\E\{e_a(t) e_b(t)\}\E\{ e_c(t) e_d(t)\}\\
&-\E\{e_a(t) e_c(t)\}\E\{ e_b(t) e_d(t)\}-\E\{e_a(t) e_d(t)\}\E\{ e_b(t) e_c(t)\}.
\end{align*}
\label{lem}
\end{lemma}

Let $\mathbf{Y}(t)=\mathbf{\Phi}^T\mathbf{X}(t)$ be the principal components of $\mathbf{X}(t)$, and $q(\omega)=\bm{\Phi}^Th(\omega)$ be its transfer function.
Then, the spectral density matrix of $\mathbf{Y}(t)$ is 
\begin{equation*}
    g(\omega)=\bm{\Phi}^Tf(\omega)\bm{\Phi}=(2\pi)^{-1}\sum_{s=-\infty}^{\infty}\exp\{-\mathrm{i}s\omega\}\mathbf{\Gamma}_{\textbf{Y}}(s),
\end{equation*}
with $\mathbf{\Gamma}_{\textbf{Y}}(s)=\E(\mathbf{Y}(t+s)\mathbf{Y}(t)^T)$ for $s \in \{...,-1,0,1,...\}$.
Although $\mathbf{\Gamma}_{\textbf{Y}}(s)$ is diagonal at $s=0$, it is not so at other $s$.
Thus, different from the PCs obtained for i.i.d. data, which is free from spatial dependence, there is temporal spatial dependence in $\mathbf{Y}(t)$.
Applying Lemma \ref{lem} to $\mathbf{Y}(t)$, Corollary \ref{coro:g} follows.

\begin{corollary}\label{coro:g}
Let $\mathbf{T}_n=\bm{\Phi}^T\mathbf{C}_n\bm{\Phi}$ and $\mathbf{U}=\sqrt{n}(\mathbf{T}_n-\bm{\Lambda})$. Denote the $(i,j)$-th component of $\mathbf{U}$ by $u_{ij}$. Then, $\{u_{ij}\}$ has a joint asymptotic normal distribution whose mean is zero and whose asymptotic covariance (denoted as $\AC$) between $\{u_{ij}\}$ and $\{u_{kl}\}$ is
\begin{align}
\begin{split}
\AC(u_{ij},u_{kl})=
&2 \pi \int_{-\pi}^{\pi}\left\{g_{i k}(\omega) \overline{g_{j l}(\omega)}+g_{i l}(\omega) \overline{g_{j k}(\omega)}\right\} d \omega\\
+&(4 \pi^2)^{-1}\sum_{a,b,c,d=1}^p\kappa_{abcd}\int_{-\pi}^{\pi}\int_{-\pi}^{\pi}q_{ia}(\omega_1)\overline{q_{jb}(\omega_1)} \overline{q_{kc}(\omega_2)}q_{ld}(\omega_2)d\omega_1 d\omega_2.
\label{eq:u}
\end{split}
\end{align}
\end{corollary}

Now, we state the main result, a central limit theorem for the eigenvalues, the eigenvectors and the proportion of variation as follows.

\begin{theorem}
Let $\mathbf{X}(t)$ satisfy the assumptions of Lemma \ref{lem} and suppose $\lambda_1>\lambda_2>...>\lambda_p>0$. 
Define $\mathbf{D}=\sqrt{n}(\mathbf{L-\Lambda})$, $\mathbf{H}=\sqrt{n}(\mathbf{A-\Phi})$ and $\mathbf{V}=\mathbf{\Phi}^T \mathbf{H}$.
Then, for the elements of the matrix $\mathbf{D}=\diag(d_i)$ and $\mathbf{V}=\{v_{ij}\}$ for $i,j=1,...,p$, 
the following results hold 
\begin{equation}\label{eq:dw}
    d_{i}=u_{ii}+o(1), \quad v_{ii}=o(1), \quad \mbox{and} \quad v_{ij}=u_{ij}/(\lambda_j-\lambda_i)+o(1), i \neq j,
\end{equation}
a.s. as $n \to \infty$.
Then, the limiting distributions of $\mathbf{D}$ and $\mathbf{H}$ are obtained since the asymptotic distribution of $\{u_{ij}\}$ is given in Corollary \ref{coro:g}.

Let $\Longrightarrow$ denote weak convergence. We show the asymptotics of the proportion of variation and loadings of PCs specifically.
\begin{enumerate}[(1)]
    \item The eigenvalues $\mathbf{l}$ has the limiting distribution $\sqrt{n}(\mathbf{l}-\bm{\lambda})\Longrightarrow \mathcal{N}(0, \mathbf{B})$, where $\mathbf{B}=\{b_{ij}\}_{i,j=1,\ldots,p}$ with $b_{ij}=\AC(u_{ii},u_{jj})$.
    
    For $k=1,...,p$, the proportion of variation ${r}_k$ has the limiting distribution
\begin{equation}\label{eq:eta}
    \sqrt{n}({r}_k-{\gamma}_k)\Longrightarrow \mathcal{N}(0,\eta_k^2),\quad 
    \eta_k^2=\sum_{i=1}^p\sum_{j=1}^p\AC(u_{ii},u_{jj})\beta_i(k)\beta_j(k),
\end{equation}
with 
$\beta_i(k)=(\mathds{1}(i\leq k)-{\gamma}_k)/\mathrm{tr}(\mathbf{\Gamma})$.
\item For $k=1,...,p$, the PC loadings $\mbf{a}_k$ has the limiting distribution
\begin{equation}\label{eq:Sigma}
    \sqrt{n}(\mbf{a}_k-\bm{\phi}_k)\Longrightarrow \mathcal{N}(0,\mbf{\Sigma}_k),\quad
    \mbf{\Sigma}_k=\sum_{i=1,i \neq k}^p \sum_{j=1,j \neq k}^p \frac{\AC(u_{ik},u_{jk})}{(\lambda_k-\lambda_i)(\lambda_k-\lambda_j)}\bm{\phi}_i\bm{\phi}_j^T.
\end{equation}
\end{enumerate}
\label{the}
\end{theorem}

\begin{proof}
Let $\mathbf{Z}=\mathbf{\Phi}^T \mathbf{A}$, then $\mathbf{T}_n=\mathbf{ZLZ}^T$, and $\mathbf{V}=\sqrt{n}(\mathbf{Z}- \mathbf{I})$.
By the same argument as \cite{anderson2003} (Page 546), substituting $\mathbf{T}_n=\mathbf{\Lambda}+(1 / \sqrt{n}) \mathbf{U}$, $\mathbf{Z}=\mathbf{I}+(1 / \sqrt{n}) \mathbf{V}$ and $\mathbf{L}=\mathbf{\Lambda}+(1 / \sqrt{n}) \mathbf{D}$ in the equations $\mathbf{T}_n=\mathbf{ZLZ}^T$ and $\mathbf{I}=\mathbf{ZZ}^T$, the result in (\ref{eq:dw}) follows.
Using the fact that  $\mathbf{H}=\mathbf{\Phi V}$ and Corollary \ref{coro:g}, the limiting distributions of $\mathbf{D}$ and $\mathbf{H}$ are obtained.
Using the delta method, the limiting distribution of ${r}_k$ is obtained by noting the fact that 
it is a differentiable function of $\mathbf{l}$.
Thus, the proof is complete.
\hfill $\square$
\end{proof}

For Gaussian processes, $\kappa_{abcd}$ is zero, and (\ref{eq:u}) is simplified to
\begin{equation}\label{eq:u_g}
    \AC(u_{ij},u_{kl})=
2 \pi \int_{-\pi}^{\pi}\left\{g_{i k}(\omega) \overline{g_{j l}(\omega)}+g_{i l}(\omega) \overline{g_{j k}(\omega)}\right\} d \omega.
\end{equation}

\begin{corollary}\label{coro:general}
Let $\mathbf{X}(t)$ be a Gaussian process that satisfies the assumptions of Theorem \ref{the}.
Then, the following asymptotic results hold for $k=1,2,...,p$.
\begin{enumerate}[(1)]
\item Jointly, $\mathbf{l}$ and $\mathbf{a}_k$ are asymptotically normally distributed and asymptotically unbiased.

\item For $\mbf{l}$, $\sqrt{n}(\mathbf{l}-\bm{\lambda})\Longrightarrow \mathcal{N}(0, \mathbf{B})$, where
$b_{ij}=
4 \pi \int_{-\pi}^{\pi}|g_{ij}(\omega)|^2 d \omega$.
Thus, (\ref{eq:eta}) is now
\begin{equation}\label{eq:eta_g}
    \sqrt{n}({r}_k-{\gamma}_k)\Longrightarrow \mathcal{N}(0,\eta_k^2),\quad\ where\;\;
    \eta_k^2=\sum_{i=1}^p\sum_{j=1}^p4 \pi \int_{-\pi}^{\pi}|g_{ij}(\omega)|^2 d \omega\beta_i(k)\beta_j(k).
\end{equation}
\item  For $\sqrt{n}(\mbf{a}_k-\bm{\phi}_k)\Longrightarrow \mathcal{N}(0,\mbf{\Sigma}_k)$, (\ref{eq:Sigma}) is now
\begin{align}\label{eq:Sigma_g}
    \mbf{\Sigma}_k=\sum_{i=1,i \neq k}^p \sum_{j=1,j \neq k}^p 
    \frac{2 \pi \int_{-\pi}^{\pi}\left\{g_{i j}(\omega) \overline{g_{kk}(\omega)}+g_{i k}(\omega) \overline{g_{kj}(\omega)}\right\} d \omega}{(\lambda_k-\lambda_i)(\lambda_k-\lambda_j)}\bm{\phi}_i\bm{\phi}_j^T.
\end{align}
\end{enumerate}
\end{corollary}

\section{Special cases for Gaussian processes}\label{sec:special}

As we have said, generally speaking, there is temporal spatial dependence in $\textbf{Y}_t$. However,
in some special cases, $\textbf{Y}_t$ has simpler dependence structure, thus leading to neater results. 
We discuss two such cases for Gaussian processes here.
\begin{enumerate}
\item\label{itema} All $\mathbf{\Gamma}_{\textbf{Y}}(s)$ are diagonal, which means $\textbf{Y}_t$ are independently evolving time series, i.e., there is only temporal dependence in $\textbf{Y}_t$. It follows that $g(\omega)$ is diagonal.
\item\label{itemb} $\mathbf{\Gamma}_{\textbf{Y}}(0)$ is diagonal and all other $\mathbf{\Gamma}_{\textbf{Y}}(s)$ are zero, which means $\textbf{Y}_t$ are independent vectors, i.e., there is no temporal or spatial dependence in $\textbf{Y}_t$.
\end{enumerate}
The first case has been studied by \citet{taniguchi1987asymptotic}, although it is difficult to envisage a  real situation where   $\mathbf{\Gamma}_{\textbf{Y}}(s)$ is diagonal for {\it {all}} $s$.
Perhaps, this case could be viewed as 
an approximation of Theorem \ref{the}, i.e., substituting $g(\omega)$ by its diagonalized version.
The second case covers the classical PCA for independent data $\mathbf{X}_t$.
It is developed for Gaussian data and has been extensively studied. See, e.g., \citet{anderson2003}. 
To be consistent, we state the asymptotic results for Gaussian processes under (\ref{itema}) and (\ref{itemb}) in Corollary \ref{coro:dag} and \ref{coro:ind}, respectively.

\begin{corollary}
Let $\mathbf{X}(t)$ be a Gaussian process that satisfies the assumptions of Theorem \ref{the}, and suppose all $\mathbf{\Gamma}_{\textbf{Y}}(s)$ are diagonal. Then, in (\ref{eq:u}), $\AC(u_{ij},u_{kl})$ is zero except when $i=j=k=l$, $i=k\neq j=l$ or $i=l\neq k=j$,
\begin{align*}
    \AC(u_{ii},u_{ii})=4 \pi \int_{-\pi}^{\pi}g_{i i}^2(\omega) d \omega,\quad
    \AC(u_{ij},u_{ij})=\AC(u_{ij},u_{ji})=2 \pi \int_{-\pi}^{\pi}g_{i i}(\omega) {g_{j j}(\omega)} d \omega,\quad i\neq j.
\end{align*}
which means the elements of $\mbf{U}$ are asymptotically independent.
Thus, the following asymptotic results hold for $k=1,2,...,p$.
\begin{enumerate}[(1)]

\item Jointly, $\mathbf{l}$ and $\mathbf{a}_k$ are asymptotically normally distributed and asymptotically unbiased.

\item For $\mbf{l}$, $\sqrt{n}(\mathbf{l}-\bm{\lambda})\Longrightarrow \mathcal{N}(0, \mathbf{B})$, where
$b_{ij}=
4 \pi \int_{-\pi}^{\pi}g_{ii}^2(\omega)d \omega\mathds{1}(i=j)$.
Thus, (\ref{eq:eta}) is now
\begin{equation*}
    \sqrt{n}({r}_k-{\gamma}_k)\Longrightarrow \mathcal{N}(0,\eta_k^2),\quad 
    where\;\; \eta_k^2=\sum_{i=1}^p 4 \pi \int_{-\pi}^{\pi}g_{ii}^2(\omega)d \omega\beta_i^2(k).
\end{equation*}
\item  For $\sqrt{n}(\mbf{a}_k-\bm{\phi}_k)\Longrightarrow \mathcal{N}(0,\mbf{\Sigma}_k)$, (\ref{eq:Sigma}) is now
\begin{align*}
    \mbf{\Sigma}_k=\sum_{i=1, i \neq k}^{p} \frac{2 \pi \int_{-\pi}^{\pi}{g}_{kk}(\omega){g}_{ii}(\omega)d \omega}{\left(\lambda_{k}-\lambda_{i}\right)^{2}}\bm{\phi}_{i} \bm{\phi}^T_{i}.
\end{align*}

\item  All of the $l_k$ are asymptotically independent of all of the $\mbf{a}_k$.
\end{enumerate}
\label{coro:dag}
\end{corollary}

\begin{corollary}
Let $\mathbf{X}(t)$ be a Gaussian process that satisfies the assumptions of Theorem \ref{the}.
Suppose $f(\omega)= \mathbf{\Gamma}/(2\pi)$, which means $\mathbf{X}(t)$ is i.i.d. and $\mathcal{N}(0,\mathbf{\Gamma})$, and it follows that $g(\omega)= \mathbf{\Lambda}/(2\pi)$.
Then, in (\ref{eq:u}), $\AC(u_{ij},u_{kl})$ is zero except when $i=j=k=l$, $i=k\neq j=l$ or $i=l\neq k=j$,
\begin{align*}
    \AC(u_{ii},u_{ii})=2 \lambda_i^2,\quad
    \AC(u_{ij},u_{ij})=\AC(u_{ij},u_{ji})=\lambda_i\lambda_j,\quad i\neq j,
\end{align*}
which means the elements of $\mbf{U}$ are asymptotically independent.
Thus, the following asymptotic results hold for $k=1,2,...,p$.
\begin{enumerate}[(1)]
\item Jointly, $\mathbf{l}$ and $\mathbf{a}_k$ are asymptotically normally distributed and asymptotically unbiased.

\item For $\mbf{l}$, $\sqrt{n}(\mathbf{l}-\bm{\lambda})\Longrightarrow \mathcal{N}(0, \mathbf{B})$, where
$b_{ij}=
2 \lambda_{i}^{2}\mathds{1}(i=j)$.
Thus, (\ref{eq:eta}) is now
\begin{equation*}
    \sqrt{n}({r}_k-{\gamma}_k)\Longrightarrow \mathcal{N}(0,\eta_k^2),\quad 
    where \;\; \eta_k^2=\sum_{i=1}^p 2 \lambda_{i}^{2}\beta_i^2(k).
\end{equation*}

\item  For $\sqrt{n}(\mbf{a}_k-\bm{\phi}_k)\Longrightarrow \mathcal{N}(0,\mbf{\Sigma}_k)$, (\ref{eq:Sigma}) is now
\begin{align*}
    \mbf{\Sigma}_k=\sum_{i=1, i \neq k}^{p} \frac{{\lambda}_{k}{\lambda}_{i}}{\left(\lambda_{k}-\lambda_{i}\right)^{2}}\bm{\phi}_{i} \bm{\phi}^T_{i}.
\end{align*}

\item  All of the $l_k$ are asymptotically independent of all of the $\mathbf{a}_k$.
\end{enumerate}
\label{coro:ind}
\end{corollary}

We compare the results of Corollary \ref{coro:general}, \ref{coro:dag} and \ref{coro:ind} for Gaussian processes in Table \ref{table:sec2}.
\begin{table}
\caption{\label{table:sec2}Comparisons among Corollary \ref{coro:general}, \ref{coro:dag} and \ref{coro:ind} for Gaussian processes. A \cmark  means yes and a \xmark  means no.}
\centering
\footnotesize
\begin{tabular}{l|lll}
\toprule
    Corollary & \ref{coro:general} & \ref{coro:dag} & \ref{coro:ind}\\\hline
     \textbf{Dependence in $\mathbf{Y}_t$ }& & &\\
    Temporal & \cmark & \cmark & \xmark\\
    Spatial& \cmark & \xmark & \xmark\\\hline
    \textbf{Asymptotic properties }& & &\\
    Joint Gaussianity of $\mathbf{l}$ and $\mathbf{a}_k$ & \cmark & \cmark & \cmark\\
    Unbiasedness of $\mathbf{l}$ and $\mathbf{a}_k$ & \cmark & \cmark & \cmark\\
    Independence among $l_k$ & \xmark & \cmark & \cmark\\
    Independence between all $l_k$ and $\mathbf{a}_k$ & \xmark & \cmark & \cmark\\
    \bottomrule
    \end{tabular}
\end{table}
Intuitively, a simpler dependence structure would bring benefits such as asymptotic independence or simpler forms of the asymptotic covariance. 
It should be noted that Corollaries \ref{coro:general}, \ref{coro:dag} and \ref{coro:ind} study the same estimate $\mathbf{l}$, ${r}_k$ and $\mathbf{A}$, but give different asymptotics. 
An essential question is how these different dependence structures affect the inference of the number and loadings of PCs in practice.
Especially, we want to know the consequences of making inference under the assumption of independent observations (Corollary \ref{coro:ind}) when the real data is time series (Corollary \ref{coro:general}). This situation is possibly the most common misuse of PCA for time series.
Thus, we will discuss the estimation of the asymptotic covariance of eigenvalues and eigenvectors in the next Section, which is essential for the inference of the number and loadings of PCs.

\section{Estimation of the asymptotic covariance of eigenvalues and eigenvectors}\label{sec:numerical}
To conduct inference on the eigenvalues, eigenvectors and proportion of variation, an intuitively plausible approach is to produce directly consistent estimates of $\mbf{B}$, $\mbf{\Sigma}_k$ and $\eta_k^2$ from  their analytical forms in (\ref{eq:eta}) and (\ref{eq:Sigma}).
The main challenge lies in the second term of $\AC(u_{ij},u_{kl})$ in (\ref{eq:u}). 
In principle, this can be handled by the estimation of the integral of the fourth-order cumulant spectra as similar challenges often appear in the asymptotic covariances of quantities in time series analysis, such as spectral mean statistics \citep{shao2009confidence}, whittle estimation \citep{giraitis2001whittle} and so on. 
Now, this was the approach adopted by \cite{taniguchi1996nonparametric}.
However, their proposed estimator turns out to be computationally complex: its consistency requires the existence of the eighth order moment
and the procedure involves a choice of a smoothing parameter with no theoretical guidance.
As a result, in  the literature, researchers have tended to avoid direct estimation of this term and preferred more appealing alternatives, such as the self-normalization approach \citep{shao2009confidence} and the bootstrap method \citep{meyer2020extending}.
Thus, we will only implement the direct estimation when $\kappa_{abcd}$ is negligible, e.g., Gaussian processes, under which case  Corollary \ref{coro:general} can be applied.
Simulation results in the next section will show that this method has some scope of application for some light-tailed data.

Our experience suggests that a data-dependent 
bootstrap method can alleviate the  troublesome aspects mentioned earlier and can be applied to general processes beyond the Gaussian.

\subsection{Direct estimation}

For a given Gaussian process, to estimate the asymptotic covariances $\mbf{B}$, $\mbf{\Sigma}_k$ and $\eta_k^2$ in Corollary \ref{coro:general}, we need to estimate $\mbf{\Phi},\mbf{\Lambda}$ and $g(\omega)$.
For $\mbf{\Phi}$ and $\mbf{\Lambda}$, natural estimators are $\mbf{A}$ and $\mbf{L}$. 
For the spectral density matrix $f(\omega)$, define the raw periodogram $\tilde{f}(\omega)$ and the smoothed periodogram $\hat{f}(\omega)$ as 
\begin{equation*}
    \tilde{f}(\omega)=(2\pi)^{-1}\sum_{s=-(n-1)}^{n-1}\exp\{-\mathrm{i}s\omega\}\hat{\mathbf{\Gamma}}(s),\quad
    \hat{f}(\omega)=\sum_{m=-M}^M W_M(m)\tilde{f}(\omega+m/n),
\end{equation*}
where $\hat{\mathbf{\Gamma}}(s)=(1/n) \sum_{t=1}^n (\mathbf{X}(t+s)-\bar{\mathbf{X}})(\mathbf{X}(t)-\bar{\mathbf{X}})^T$, $M$ is the bandwidth, and $W_M(m)$ is the spectral window function.
For example, $W_M(m)=\mathds{1}(|m|\leq M)/(2M+1)$ is the Daniell window.
It is known that $\hat{f}(\omega)$ is a consistent estimator. 
Thus, the estimator of $g(\omega)$ is $\hat{g}(\omega)=\mbf{A}^T\hat{f}(\omega)\mbf{A}$. 
Then, we obtain the following estimators of $\mbf{B}$, $\mbf{\Sigma}_k$ and $\eta_k^2$,
\begin{align*}
    &\hat{b}_{ij}=4 \pi \int_{-\pi}^{\pi}|\hat{g}_{ij}(\omega)|^2,\quad
    \hat{\eta}_k^2 =\sum_{i=1}^p\sum_{j=1}^p4 \pi \int_{-\pi}^{\pi}|\hat{g}_{ij}(\omega)|^2d \omega \hat{\beta}_i(k)\hat{\beta}_j(k),\\
    &\hat{\mbf{\Sigma}}_k=\sum_{i=1,i \neq k}^p  \sum_{j=1,j \neq k}^p 
    \frac{2 \pi \int_{-\pi}^{\pi}\left\{\hat{g}_{i j}(\omega) \overline{\hat{g}_{kk}(\omega)}+\hat{g}_{i k}(\omega) \overline{\hat{g}_{kj}(\omega)}\right\} d \omega}{(l_k-l_i)(l_k-l_j)}\mbf{a}_i\mbf{a}_j^T,
\end{align*}
with $\hat{\beta}_i(k)=(\mathds{1}(i\leq k)-{r}_k)/\sum_{j=1}^p l_j$.

Similar procedures can be applied to obtain the direct estimation of the asymptotic covariances $\mbf{B}$, $\mbf{\Sigma}_k$ and $\eta_k^2$ in Corollary \ref{coro:dag} and \ref{coro:ind}.

\subsection{Bootstrap method}

Our second method is to resample and derive the bootstrap covariance estimator, avoiding the difficulty and complexity of estimating the second term in (\ref{eq:u}). 
It is applicable to general processes, especially non-Gaussian and heavy-tailed processes. 
Various bootstrap methods for dependent data have been explored by researchers, including block bootstrap, sieve bootstrap, bootstrap in frequency domain, and others. Comprehensive reviews of bootstrap for dependent data can be found in \cite{buhlmann2002}, \cite{lahiri2003} and \cite{politis2003}.

Our experiences with the surrogate data method, the time frequency toggle method and the moving block bootstrap have narrowed the choice to the last method. The first two methods are in the frequency domain.
The idea of the surrogate data method from \cite{theiler1992} is to bootstrap the phase of the Fourier coefficients but keep their magnitude unchanged. 
This method has limited validity because every surrogated sample has exactly the same periodogram (and mean) as the original series. 
For this reason, this method fails for our problem since it will generate exactly the same eigenvectors.

The time frequency toggle method was proposed by \cite{kirch2011}, whose basic idea is to bootstrap the Fourier coefficients of the observed time series, and then to back-transform them to obtain a bootstrap sample in the time domain. 
It is more general than the surrogate data method in that it can capture the distribution of statistics based on the periodogram, and thus can work for our problem. 
However, our simulation results show that its performance for heavy-tailed data is not good. 
The reason may lie in the fact that the sample paths of this method are asymptotically Gaussian, due to the Fourier coefficients being asymptotically so. This seems inevitable for almost all methods using discrete Fourier transforms.

The moving block bootstrap (MBB) originated in \cite{hall1985} and \cite{carlstein1986}, and was further developed in \cite{kunsch1989}, \cite{liu1992} and others for stationary observations.
It resamples blocks of (consecutive) observations at a time. 
As a result, the dependence structure of the original observations is preserved within each block. 
To explain briefly, let us consider an  observed univariate time series $X(1),...,X(n)$,
for a parameter $\theta$ which is a functional of the $j$-dimensional marginal distribution of the time series, we consider the $j$-dimensional vectors
\[
\mbf{S}(t)=(X(t),...,X(t+j-1)), \quad t=1,...,n-j+1.
\]
For fixed block size $\rho$, define the blocks in terms of $\mbf{S}(t)$ as
\[\Xi_t=(\mbf{S}(t),...,\mbf{S}(t+\rho-1)),\quad t=1,..,n+2-\rho-j.\]
For $\tau \geq 1$, select $\tau$ blocks randomly from the collection $\{\Xi_t,t=1,..,n+2-\rho-j \}$, and align them to generate the MBB observations 
$\{\mbf{S}^*(1),...\mbf{S}^*(\rho);\mbf{S}^*(\rho+1),...,\mbf{S}^*(2\rho);...,\mbf{S}^*(\tau\rho)\}$. 
And $\tau$ should be the smallest integer such that $\tau\rho \geq n$.
Then, the estimator $\hat{\theta}$ is derived based on the MBB observations.

It can be seen that the above procedure can be easily modified for multivariate time series, so it can be applied to our PCA problem.
However, the eigenvectors depend on the entire distribution of the process, so the question of decision of $j$ arises.
An ad-hoc solution is to work with the naive block bootstrap (using $j = 1$). 
This will inevitably result in some efficiency loss compared to the ordinary block bootstrap, but the simulation in the next section suggests that the numerical performance is good and acceptable. 
So, we will use MBB for our problem, while also leaving the door open for other bootstrap methods.

\section{Simulation}\label{sec:simu}

In this section, we will assess, via simulation,
the asymptotics and estimation methods in above sections.
Instead of focusing on the eigenvalues, eigenvectors and the proportion of variation ($\bm{\lambda}$, $\mbf{\Phi}$ and $\bm{\gamma}$), 
our major interest is here in the asymptotic covariances of their estimates ($\mbf{l}$, $\mbf{A}$ and $\mbf{r}$), namely $\mbf{B}$, $\mbf{\Sigma}_k$ and $\eta_k^2$,
with different formulas given by Corollaries \ref{coro:general}, \ref{coro:dag} and \ref{coro:ind} for Gaussian processes.
In Section \ref{subsec1}, we will check the efficacy of these formulas.
The asymptotic distribution of Corollary \ref{coro:general} is denoted by AD.
Results based on Corollaries \ref{coro:dag} and \ref{coro:ind} are abbreviated as DAG and IND to indicate the case with diagonalized spectral density matrix and the case of independent observations, respectively.
Conducting PCA to simulated data and repeating over replications, we can get their empirical distribution, which is denoted as ED. 
By comparing the formula of $\mbf{B}$, $\mbf{\Sigma}_k$ and $\eta_k^2$ given be AD, DAG and IND respectively with the empirical covariance of $\sqrt{n}(\mathbf{l}-\bm{\lambda})$, $\sqrt{n}(\mbf{a}_k-\bm{\phi}_k)$ and $\sqrt{n}({r}_k-{\gamma}_k)$ given by ED, we can assess their performance.

In Section \ref{subsec2}, we assess the two estimation methods, i.e. the direct method (DE) and the bootstrap method (BE) introduced in Section \ref{sec:numerical}.
For each simulation replication, we estimate the standard deviation of eigenvectors using the two estimation methods and averaging over replications.
We compare them with the empirical standard deviations (ED).
Here we handle not only Gaussian processes, but also non-Gaussian ones, including some with outliers, or skewed Gaussianity, or heavy-tail.

In the following simulation, we take finite-order Vector Autoregressive (VAR) and Vector Moving Average models (VMA) as examples, which are respectively defined as
\begin{equation*}
    \mathbf{X}(t)=\mathbf{e}(t)+\sum_{j=1}^J\mathbf{F}(j) \mathbf{X}(t-j), \quad
    \mathbf{X}(t)=\mathbf{e}(t)+\sum_{j=1}^J\mathbf{G}(j) \mathbf{e}(t-j),
\end{equation*}
where $\mathbf{e}(t)$ is a white noise series and $J$ is the order.

\subsection{Efficacy of the asymptotics for Gaussian processes}\label{subsec1}

Let the length of time series be $n=5000$, the dimension be $p=5$, and the replication time be $N=2000$.
Suppose $\mathbf{e}(t)$ is from $\mathcal{N}({0}_5,10{\mathbf{I}}_5)$, with $\mathbf{I}_5$ being the identity matrix of dimension 5.
Consider the following Data Generating Processes (DGPs):

\begin{itemize}
\item \textbf{DGP 1}: VAR(1) process with multivariate Gaussian noise.

\item \textbf{DGP 2}: VMA(1) process with multivariate Gaussian noise.

\end{itemize}

Elements of $\mathbf{F}(j)$ and $\mathbf{G}(j)$ are obtained by first drawing random numbers from the uniform distribution and then transformed in order to ensure the stationarity of the VAR models and  invertibility of the VMA models. 
And they are shown in Appendix \ref{subsec:app1}.

For the eigenvalues we will check their asymptotic covariance matrix $\Cov\left(l_{k}, l_{k'}\right)\equiv b_{kk'}/n$ to show the asymptotic dependence among them.
For the eigenvectors and the proportion of variation, we will check the asymptotic standard deviation of each element, $\sigma({a}_{kk'})\equiv {\Sigma}_{k,k'k'}/\sqrt{n}$ and $\sigma({r}_k)\equiv\eta_k/\sqrt{n}$, for the sake of inference.
We check $\sigma({r}_k)$ for $k=1,\ldots,p-1$ since $r_p$ is always 1 and thus $\sigma({r}_p)$ is always 0.
Note that they all have an empirical entry (ED), and three asymptotic entries (AD, DAG and IND), which will be marked in the subscript.

For $\sigma({r}_k)$, we compute the difference ($\Delta$) of AD with respect to ED in percentage as
\begin{equation}\label{eq:deltar}
\Delta_{\mathrm{{AD}}}({r}_k)=\sigma_{\mathrm{AD}}({r}_k)-\sigma_{\mathrm{ED}}({r}_k).
\end{equation}
For $\sigma({a}_{kk'})$, to avoid the influence of scale, we compute the difference ratio ($\Delta^*$) of AD with respect to ED in percentage as
\begin{equation}
\Delta^*_{\mathrm{{AD}}}({a}_{kk'})=\sigma_{\mathrm{AD}}({a}_{kk'})/\sigma_{\mathrm{ED}}({a}_{kk'})-1.
\label{eq:deltaa}
\end{equation}
Similarly, we can compute $\Delta_{\mathrm{{DAG}}}({r}_k)$, $\Delta_{\mathrm{{IND}}}({r}_k)$, $\Delta^*_{\mathrm{{DAG}}}({a}_{kk'})$ and $\Delta^*_{\mathrm{{IND}}}({a}_{kk'})$.

For each model, the vectorized $\Cov\left(l_{k}, l_{k'}\right)$ and $\Delta^*({a}_{kk'})$ for $k,k'=1,...,5$ are each a $25$-dimensional vector. 
See their plots in Fig.\ref{fig3} for DGP 1 and 2 respectively.
\begin{figure}
\centering
\makebox{\includegraphics[width=0.8\textwidth]{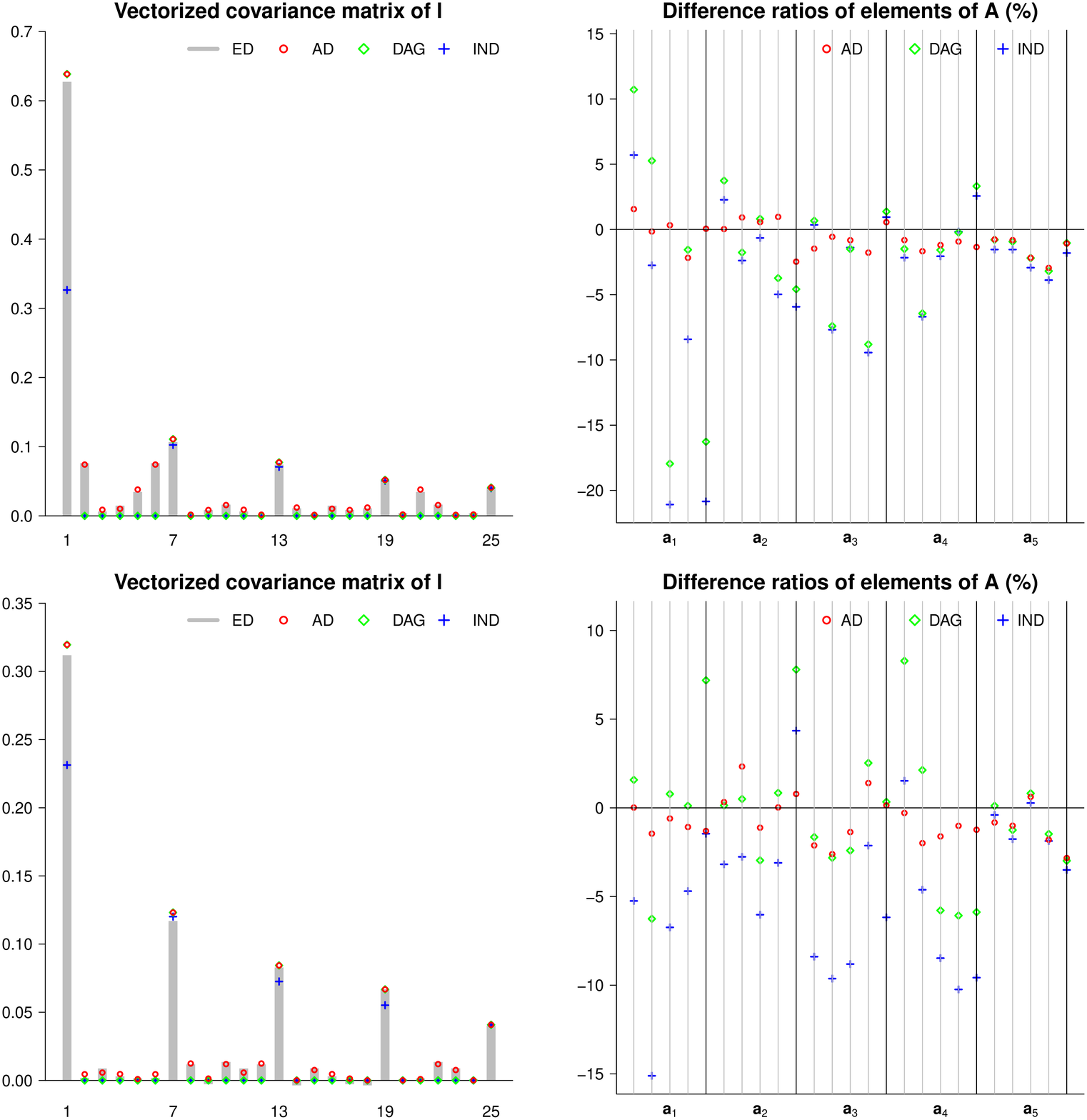}}
\caption{The vectorized $\Cov\left(l_{k}, l_{k'}\right)$ and $\Delta^*({a}_{kk'})$ for $k,k'=1,...,5$ for DGP 1 (Top pane) and DGP 2 (Bottom pane).} 
\label{fig3}
\end{figure}
For the vectorized $\Cov\left(l_{k}, l_{k'}\right)$, the 1st, 7th, 13rd, 19th and 25th elements are $\Var\left(l_{k}\right)\equiv\Cov\left(l_{k}, l_{k}\right)$, and other elements are $\Cov\left(l_{k}, l_{k'}\right)$ for $k \neq k'$.  
From the figure, we can see AD fits ED well. IND tends to give a smaller $\Var\left(l_{k}\right)$, while DAG gives the same $\Var\left(l_{k}\right)$ as AD, both cases being consistent with the theoretical results.
However, for $\Cov\left(l_{k}, l_{k'}\right), k \neq k'$, IND and DAG claim that $l_{k}$ are uncorrelated, which is not supported by ED.

For $\sigma({a}_{kk'})$, AD performs pretty well and the difference ratios are controlled within 5\%. In contrast, IND and DAG have poor performance; the difference ratio can even reach a figure as big as 20\% for DGP 1.

As for the proportion of variation, Table \ref{table:prop} shows $\Delta({r}_k)$ subscripted by AD, DAG and IND.
\begin{table}
\caption{\label{table:prop}The difference $\Delta({r}_k)$ in percentage for $k=1,...,4$ for DGP 1 and DGP 2.}
\centering
\footnotesize
\begin{tabular}{rrrrr|rrrr}
  \hline
   & \multicolumn{4}{c|}{DGP1} & \multicolumn{4}{c}{DGP2}\\
 $k$ & 1 & 2 & 3 & 4 & 1 & 2 & 3 & 4 \\
  \hline
AD & -0.015 & -0.011 & -0.008 & -0.000 & 0.009 & 0.010 & -0.002 & 0.008 \\ 
DAG & 0.045 & -0.009 & 0.031 & 0.027 & 0.009 & 0.030 & 0.013 & 0.020 \\ 
IND  & -0.125 & -0.108 & -0.018 & 0.010 & -0.063 & -0.016 & -0.018 & 0.013 \\ 
   \hline
\end{tabular}
\end{table}
Among the three asymptotic standard deviations, $\sigma_{\mathrm{AD}}({r}_k)$ is most close to $\sigma_{\mathrm{ED}}({r}_k)$, with the differences controlled within 0.02\%.
Although the differences for $\sigma_{\mathrm{DAG}}({r}_k)$ and $\sigma_{\mathrm{IND}}({r}_k)$ are larger, they are all controlled within 0.2\%, which is negligible.

We have implemented two additional  simulation DGPs, which is of higher order. See the Appendix \ref{subsec:app2}.
\begin{itemize}
    \item \textbf{DGP 3}: VMA(2) model with multivariate Gaussian noise.
    \item \textbf{DGP 4}: VMA(3) model with multivariate Gaussian noise.
\end{itemize}
To summarize, for Gaussian processes, the simulation has confirmed the efficacy of Theorem \ref{the} and indicated the deficiency of the diagonal spectral density matrix (DAG) assumption and the independent-data (IND) assumption for time series.
On the positive side, the standard deviations of the proportion of variation have negligible difference under the three assumptions.
Thus, making interval inference under which assumption makes not much difference for determining the number of effective PCs.
However, DAG and IND perform poorly for the eigenvectors, which can seriously affect the interpretation of PC loadings.

\subsection{Efficacy of the estimation methods}\label{subsec2}

In this part, we pay attention to the estimation of the standard deviation of the eigenvectors and proportions of variation, namely $\sigma({a}_{kk'})$ and $\sigma({r}_k)$.
For each simulated sample of time length $n=2000$ and dimension $p=5$, we estimate the standard deviation by direct estimation. 
We average over $N=2000$ samples and denote the average standard deviation as $\sigma_{\mathrm{DE}}$.
For each sample, we implement bootstrap with the MBB method. 
By reference to \cite{lahiri2003}, which suggested that the optimal block size is of order $n^{1/3}$, we take the block size as $10$, and the block number is $200$ accordingly.
Then, the bootstrap estimation of standard deviation of the eigenvectors are obtained with $500$ bootstrap replications. We take average over $N=2000$ samples and denote the average standard deviation as $\sigma_{\mathrm{BE}}$.

Here we handle not only Gaussian processes, but also some non-Gaussian processes, including one process with
outliers, one skewed process and two heavy-tailed processes. The models are as follows. 
\begin{itemize}
\item \textbf{DGP 5}: VMA(1) model with multivariate Gaussian noise, but 1\% of all the noises are outliers from other multivariate Gaussian distributions.

\item \textbf{DGP 6}: VMA(1) model with noise from a multivariate skew normal distribution.

\item \textbf{DGP 7}: VMA(1) model with noise from a multivariate $t(5)$ distribution.

\item \textbf{DGP 8}: VMA(1) model with noise from a multivariate $t(8)$ distribution.
\end{itemize}
The coefficient matrix of the VMA(1) model is the same as that of DGP 2, check it and other details about the DGPs in Appendix \ref{subsec:app2}.

For the proportion of variation, similar to (\ref{eq:deltar}), we can compute $\Delta_{\mathrm{DE}}({r}_{k})$ and $\Delta_{\mathrm{BE}}({r}_{k})$ with respect to ED for $k=1,...,p$, which are shown in Table \ref{tab:madr.r} for DGPs 1-8.
It can be seen that the bootstrap estimation behaves well except that it may be slightly  sensitive to the outliers (DGP 5).
As for the direct estimation, its behavior for heavy-tailed processes (DGP 7 and 8) are worse than the bootstrap estimation.
Nevertheless, all the values in the table are controlled within 0.5\%, which is quite small in practice.

\begin{table}
\caption{\label{tab:madr.r}The $\Delta_{\mathrm{DE}}({r}_{k})$ and $\Delta_{\mathrm{BE}}({r}_{k})$ (\%) .}
\centering
    \begin{tabular}{lrrrrrrrr}
    \toprule
          & \multicolumn{2}{c}{DGP1} & \multicolumn{2}{c}{DGP2} & \multicolumn{2}{c}{DGP3} & \multicolumn{2}{c}{DGP4} \\\hline
          $k$ & \multicolumn{1}{c}{DE} & \multicolumn{1}{c}{BE} & \multicolumn{1}{c}{DE} & \multicolumn{1}{c}{BE} & \multicolumn{1}{c}{DE} & \multicolumn{1}{c}{BE} & \multicolumn{1}{c}{DE} & \multicolumn{1}{c}{BE} \\
    \midrule
    1 & -0.013 & -0.047 & 0.030 & 0.011 & 0.018 & 0.016 & 0.004 & -0.024 \\ 
  2 & 0.021 & -0.014 & 0.013 & -0.000 & 0.019 & 0.010 & 0.021 & -0.007 \\ 
  3 & 0.006 & -0.010 & 0.014 & -0.009 & -0.002 & -0.003 & 0.005 & -0.006 \\ 
  4 & 0.017 & 0.003 & 0.012 & -0.000 & -0.001 & -0.000 & 0.007 & -0.004 \\ 
    \midrule
          & \multicolumn{2}{c}{DGP5} & \multicolumn{2}{c}{DGP6} & \multicolumn{2}{c}{DGP7} & \multicolumn{2}{c}{DGP8} \\\hline
          $k$ & \multicolumn{1}{c}{DE} & \multicolumn{1}{c}{BE} & \multicolumn{1}{c}{DE} & \multicolumn{1}{c}{BE} & \multicolumn{1}{c}{DE} & \multicolumn{1}{c}{BE} & \multicolumn{1}{c}{DE} & \multicolumn{1}{c}{BE} \\
    \midrule
    1 & 0.035 & 0.252 & -0.077 & -0.015 & -0.428 & -0.088 & -0.147 & -0.041 \\ 
  2 & 0.016 & 0.182 & -0.049 & -0.012 & -0.363 & -0.067 & -0.118 & -0.022 \\ 
  3 & 0.026 & 0.137 & -0.023 & -0.010 & -0.249 & -0.028 & -0.078 & -0.010 \\ 
  4 & 0.006 & 0.044 & -0.006 & -0.006 & -0.175 & -0.038 & -0.071 & -0.020 \\
    \bottomrule
    \end{tabular}%
\end{table}%

As for the loadings, similar to (\ref{eq:deltaa}), we can compute $\Delta^*_{\mathrm{DE}}({a}_{kk'})$ and $\Delta^*_{\mathrm{BE}}({a}_{kk'})$ with respect to ED for $k,k'=1,...,p$. 
Furthermore, to condense and summarize the element-wise information, we compute the mean absolute difference ratio ($\widetilde{\Delta}$) for each eigenvector as
\[
\widetilde{\Delta}_{\mathrm{DE}}(\mbf{a}_{k})=\sum_{k'=1}^p |\Delta^*_{\mathrm{DE}}({a}_{kk'})|/p, \quad 
\widetilde{\Delta}_{\mathrm{BE}}(\mbf{a}_{k})=\sum_{k'=1}^p |\Delta^*_{\mathrm{BE}}({a}_{kk'})|/p,
\]
which can represent the accuracy of the standard deviation of the $k$th eigenvector. The $\widetilde{\Delta}$s (\%) of PC loadings for DGPs 1-8 are shown in Table \ref{tab:madr}.
\begin{table}
\caption{\label{tab:madr}The $\widetilde{\Delta}_{\mathrm{DE}}(\mbf{a}_{k})$ and $\widetilde{\Delta}_{\mathrm{BE}}(\mbf{a}_{k})$ (\%).}
\centering
    \begin{tabular}{lrrrrrrrr}
    \toprule
          & \multicolumn{2}{c}{DGP1} & \multicolumn{2}{c}{DGP2} & \multicolumn{2}{c}{DGP3} & \multicolumn{2}{c}{DGP4} \\\hline
          & \multicolumn{1}{c}{DE} & \multicolumn{1}{c}{BE} & \multicolumn{1}{c}{DE} & \multicolumn{1}{c}{BE} & \multicolumn{1}{c}{DE} & \multicolumn{1}{c}{BE} & \multicolumn{1}{c}{DE} & \multicolumn{1}{c}{BE} \\
    \midrule
    PC1 & 1.72 & 1.83 & 1.83 & 1.78 & 0.77 & 2.49 & 2.17 & 4.63 \\ 
  PC2 & 1.40 & 4.83 & 1.04 & 1.93 & 1.35 & 1.78 & 3.77 & 10.06 \\ 
  PC3 & 2.44 & 7.52 & 2.10 & 8.92 & 1.01 & 1.73 & 0.82 & 4.72 \\ 
  PC4 & 7.32 & 10.38 & 9.36 & \textbf{17.58} & 1.01 & 1.79 & 2.00 & 4.47 \\ 
  PC5 & 4.87 & 5.10 & 7.77 & 11.07 & 1.13 & 1.15 & 1.32 & 0.88 \\
    \midrule
          & \multicolumn{2}{c}{DGP5} & \multicolumn{2}{c}{DGP6} & \multicolumn{2}{c}{DGP7} & \multicolumn{2}{c}{DGP8} \\\hline
          & \multicolumn{1}{c}{DE} & \multicolumn{1}{c}{BE} & \multicolumn{1}{c}{DE} & \multicolumn{1}{c}{BE} & \multicolumn{1}{c}{DE} & \multicolumn{1}{c}{BE} & \multicolumn{1}{c}{DE} & \multicolumn{1}{c}{BE} \\
    \midrule
    PC1 & 1.89 & 9.11 & 1.34 & 1.44 & \textbf{29.02} & 3.19 & 11.70 & 2.07 \\ 
  PC2 & 1.86 & 8.42 & 1.51 & 1.33 & \textbf{29.66} & 1.36 & 11.11 & 1.81 \\ 
  PC3 & 3.19 & 12.22 & 1.25 & 2.33 & \textbf{32.39} & 6.07 & 14.65 & 4.74 \\ 
  PC4 & \textbf{17.45} & 13.09 & 2.34 & 7.82 & \textbf{30.02} & 12.08 & 13.47 & 9.84 \\ 
  PC5 & 1.12 & 7.12 & 1.45 & 9.46 & \textbf{36.47} & 11.00 & \textbf{19.48} & 5.08 \\
    \bottomrule
    \end{tabular}%
\end{table}%
From the table, it can be seen that for most models, the $\widetilde{\Delta}$s are controlled within 15\%. Although some $\widetilde{\Delta}$s exceed 15\%, such as in DGP 2 and 5, they occur in the less consequential 4th or 5th PC. 
For DGPs 1 to 6, most $\widetilde{\Delta}$s of DE are smaller than those of BE, meaning that the direct estimation method outperforms the bootstrap method. 
These models are Gaussian, skew normal or Gaussian with outliers, and are or are close to a Gaussian process. This fact might explain the success of the direct estimation method.

However, for the heavy-tailed processes in DGP 7 and 8, the direct estimation method behaves poorly. 
The tail of DGP 7 is heavier than DGP 8, and the performance is worse.
The poor performance for the heavy-tailed process is due to the joint fourth cumulant of heavy-tailed processes being non-negligible. 
Thus, we recommend the bootstrap method, which can be implemented for general processes easily. 
\section{An empirical example}\label{sec:eg}

Now, we return to principal portfolio management by studying an empirical example, the PCA of which plays a central role,  so as
to illustrate the importance of Theorem \ref{the} and the utility of the bootstrap method, and to sound a warning against cavalier usage of PCA in time series without due attention to the standard errors of the estimated eigenvectors.
In doing so, our hope is that our procedure will provide a paradigm of implementing PCA for time series that encompasses estimation, inference and interpretation.

We use a set of the daily returns of 10 stocks traded in the New York Stock Exchange.  It consists of CVX (Chevron), XOM (Exxon), AAPL (Apple), FB (Facebook), MSFT (Microsoft), MRK (Merck), PFE (Pfizer), BAC (Bank of America), JPM (JP Morgan), and WFC (Wells Fargo \& Co.). The data are from August 2, 2016 to December 30, 2016, a total of 106 days. Within the 10 stocks, CVX and XOM belong to the energy industry, while AAPL, FB and MSFT are in the information technology (IT) sector. MRK and PFE are both health companies. BAC, JPM, and WFC are in the financial sector. See, e.g., Section 4.5 of \cite{wei2018} for a fuller description of the data set.  

First, we describe the data and  check informally whether the data satisfy the basic assumptions: ergodicity, stationarity and Gaussianity. From the marginal distribution, the time series plot in Fig.\ref{figd10} and the QQ plot of the normalized data in Fig.\ref{figqq}, we can see that the time series do not appear to display obvious trend non-stationarity apart from a few possible outliers. The data appear to be non-Gaussian, and there are heavy tails. 
Thus, we consider using the bootstrap method to estimate the standard deviations, based on which we can make inference on the number and loadings of the PCs.

\begin{figure}[!htbp]
\centering
\makebox{\includegraphics[width=0.8\textwidth]{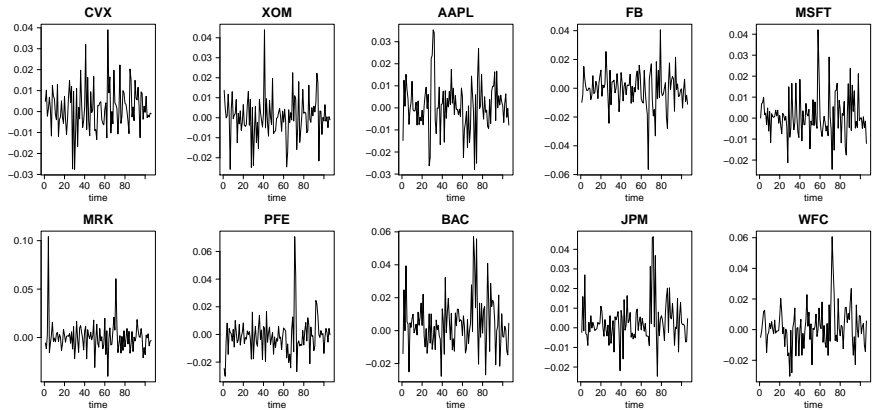}}
\caption{The time series plot of the stock returns data.} 
\label{figd10}
\end{figure}

\begin{figure}[!htbp]
\centering
\makebox{\includegraphics[width=0.8\textwidth]{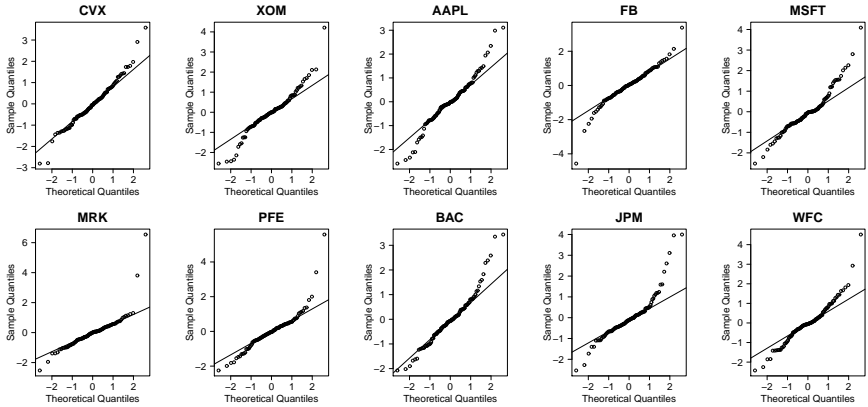}}
\caption{The QQ plot of the normalized stock returns data.} 
\label{figqq}
\end{figure}

In Fig.\ref{figpro}, we plot the proportion of variation and its 95\% confidence interval with the estimated standard deviation by the bootstrap method.
The eigenvalues of the first four components account for $r_4=77.8\%$ of the total sample variance. This value together with a visual inspection of Fig.\ref{figpro} suggests the choice of four PCs.
(In \cite{wei2018}, the first four components account for 77.5\%, and the small difference is due to a recording error of the fourth eigenvalue in \cite{wei2018}. He recorded 0.00011 but re-running his R-code we obtained 0.00013. Our value, but not his, is consistent with Figure 4.2 in the book.)

\begin{figure}[!htbp]
\centering
\makebox{\includegraphics[width=0.6\textwidth]{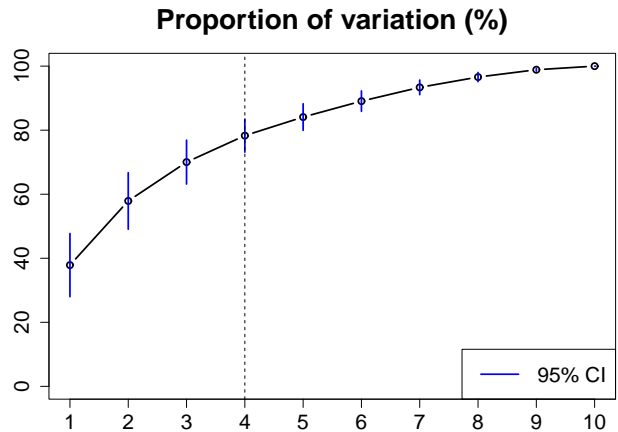}}
\caption{The proportion of variation with 95\% confidence interval.} 
\label{figpro}
\end{figure}

The PC loadings are shown in Table \ref{table2}.
The dash-lines in the tables separate the sectors.
\begin{table}
\caption{\label{table2}The PC loadings.}
\scriptsize
\begin{tabular}{rrr:rrr:rr:rrr}
  \hline
 & CVX & XOM & AAPL & FB & MSFT & MRK & PFE & BAC & JPM & WFC \\ 
  \hline
PC1 & 0.166 & 0.230 & 0.028 & 0.088 & 0.075 & 0.470 & 0.369 & 0.539 & 0.391 & 0.324 \\ 
  PC2 & 0.104 & 0.126 & 0.431 & 0.513 & 0.362 & 0.329 & 0.167 & -0.289 & -0.175 & -0.379 \\ 
  PC3 & 0.407 & 0.243 & 0.173 & 0.319 & 0.286 & -0.596 & -0.316 & 0.191 & 0.088 & 0.251 \\ 
  PC4 & 0.576 & 0.552 & -0.354 & -0.231 & -0.197 & 0.035 & 0.126 & -0.222 & -0.091 & -0.265 \\ 
  PC5 & 0.120 & 0.026 & 0.180 & -0.324 & 0.127 & 0.396 & -0.725 & 0.279 & 0.012 & -0.266 \\ 
  PC6 & 0.034 & 0.203 & 0.377 & -0.213 & -0.116 & 0.223 & -0.118 & -0.438 & -0.206 & 0.678 \\ 
  PC7 & 0.016 & -0.088 & -0.639 & 0.505 & 0.052 & 0.320 & -0.358 & -0.164 & -0.065 & 0.255 \\ 
  PC8 & 0.099 & -0.023 & 0.264 & 0.399 & -0.841 & 0.024 & -0.136 & 0.155 & -0.009 & -0.097 \\ 
  PC9 & 0.622 & -0.686 & 0.069 & -0.093 & 0.011 & 0.065 & 0.080 & -0.217 & 0.265 & 0.041 \\ 
  PC10 & 0.227 & -0.215 & -0.042 & -0.033 & 0.051 & 0.031 & 0.165 & 0.416 & -0.827 & 0.110 \\ 
   \hline
\end{tabular}
\end{table}
For every element of the first four PC loadings, $a_{kk'}$ for $k=1,\ldots,4, k'=1,\ldots, 10$, we test the null hypothesis $H_0: a_{kk'}=0$ with the estimated standard deviations in Theorem \ref{the} by the bootstrap method.
To compare, we conduct similar tests under the DAG assumption (Corollary \ref{coro:dag}) and under the IND assumption (Corollary \ref{coro:ind}), using the estimated standard deviations by the direct estimation method.
The PC loadings and inference are shown in Table \ref{table3}, where insignificant elements are suppressed (at 5\% significant level), and significant elements are shown with their signs.
The values in Table \ref{table2} are the same as that shown in \cite{wei2018} (after reversing the signs of some loadings to be consistent with our results). 
He did not conduct inference to the loadings but did a rough screening by an arbitrary  truncation in that loadings smaller than 0.1 (in absolute values) are suppressed, which are also shown in Table \ref{table3}. 

\begin{table}
\caption{\label{table3}The PC loadings with inference from Theorem \ref{the}, DAG (Corollary \ref{coro:dag}), IND(Corollary \ref{coro:ind}), and truncation in \cite{wei2018}.}
\scriptsize
\centering
    \begin{tabular}{cccc:ccc:cc:ccc}
    \toprule
          &       & CVX   & XOM   & AAPL  & FB    & MSFT  & {MRK} & PFE   & BAC   & {JPM} & WFC \\
    \midrule
    {Theorem \ref{the}} & PC1   & \textbf{+}     & \textbf{+}     &       &       &       & \textbf{+} & \textbf{+}     & \textbf{+}     & \textbf{+} & \textbf{+} \\
          & PC2   &       &       & \textbf{+}     & \textbf{+}     & \textbf{+}     &       &       &       & \large{\textbf{--}} & \large{\textbf{--}} \\
          & PC3   &       &       &       &       & \textbf{+}     & \large{\textbf{--}} & \large{\textbf{--}}     &       &       &  \\
          & PC4   & \textbf{+}     & \textbf{+}     &       &       &       &       &       &       &       &  \\
    \midrule
    {DAG} & PC1   & \textbf{+}     & \textbf{+}     &       &       &       & \textbf{+} & \textbf{+}     & \textbf{+}     & \textbf{+} & \textbf{+} \\
     (Corollary \ref{coro:dag})     & PC2   &       &       & \textbf{+}     & \textbf{+}     & \textbf{+}     & \textbf{+} &       & \large{\textbf{--}}     & \large{\textbf{--}} & \large{\textbf{--}} \\
          & PC3   & \textbf{+}     &       &       & \textbf{+}     & \textbf{+}     & \large{\textbf{--}} & \large{\textbf{--}}     & \textbf{+}     &       & \textbf{+} \\
          & PC4   & \textbf{+}     & \textbf{+}     & \large{\textbf{--}}     &       &       &       &       &      &       &  \\
    \midrule
    {IND} & PC1   & \textbf{+}     & \textbf{+}     &       &       &       & \textbf{+} & \textbf{+}     & \textbf{+}     & \textbf{+} & \textbf{+} \\
     (Corollary \ref{coro:ind})     & PC2   &       &       & \textbf{+}     & \textbf{+}     & \textbf{+}     & \textbf{+} & \textbf{+}     & \large{\textbf{--}}     & \large{\textbf{--}} & \large{\textbf{--}} \\
          & PC3   & \textbf{+}     &       &       & \textbf{+}     & \textbf{+}     & \large{\textbf{--}} & \large{\textbf{--}}     & \textbf{+}     &       & \textbf{+} \\
          & PC4   & \textbf{+}     & \textbf{+}     & \large{\textbf{--}}    &       &       &       &       & \large{\textbf{--}}     &       &  \\
    \midrule
    {Truncation } & PC1   & \textbf{+}     & \textbf{+}     &       &       &       & \textbf{+} & \textbf{+}     & \textbf{+}     & \textbf{+} & \textbf{+} \\
     \citep{wei2018}     & PC2   & \textbf{+}     & \textbf{+}     & \textbf{+}     & \textbf{+}     & \textbf{+}     & \textbf{+} & \textbf{+}     & \large{\textbf{--}}     & \large{\textbf{--}} & \large{\textbf{--}} \\
          & PC3   & \textbf{+}     & \textbf{+}     & \textbf{+}     & \textbf{+}     & \textbf{+}     & \large{\textbf{--}} & \large{\textbf{--}}     & \textbf{+}     &       & \textbf{+} \\
          & PC4   & \textbf{+}     & \textbf{+}     & \large{\textbf{--}}     & \large{\textbf{--}}     & \large{\textbf{--}}     &       & \textbf{+}     & \large{\textbf{--}}     &       & \large{\textbf{--}} \\
    \bottomrule
    \end{tabular}%
\end{table}%

As we can see from the significant loadings in Table \ref{table3}, our Theorem \ref{the} leads to the fewest number of significant elements while  the arbitrary truncation in \cite{wei2018} gives the largest number.
The results of DAG and IND are almost identical, with DAG making minimal impact.  The IND is of particular interest because it represents to-date the most frequently (mis-)used PCA for time series data, i.e., ignoring the dependence and non-Gaussianity of the observed data.
Now, the significant loadings in Table \ref{table3} lead to different interpretations of the first four PCs, as summarized in Table \ref{table4}. 
\begin{table}
\caption{\label{table4}Interpretations of the first four PCs based on Theorem \ref{the} and IND (Corollary \ref{coro:ind}).}
\centering
\begin{tabular}{rccc}
  \hline
 &\tabincell{c}{Theorem \ref{the}} &\tabincell{c}{IND (Corollary \ref{coro:ind})}&\tabincell{c}{ \cite{wei2018}}\\\hline

PC1 &\tabincell{c}{The general market\\ other than IT.} & 
\tabincell{c}{The general market\\ other than IT.}& 
\tabincell{c}{The general market\\ other than IT.}\\\hline

{PC2} &\tabincell{c}{Financial vs. {IT}.} & 
\tabincell{c}{Financial vs. {IT \& health}.}& 
\tabincell{c}{Financial vs. {non-financial}.}\\\hline

{PC3} &\tabincell{c}{Health vs. {IT}.} & 
\tabincell{c}{Health vs. {non-health}.}& 
\tabincell{c}{Health vs. {non-health}.}\\\hline

{PC4} &\tabincell{c}{ {Energy}.} & 
\tabincell{c}{ {Energy} vs. {IT \& financial}.}& 
\tabincell{c}{ {Energy} vs. {non-energy}.}\\\hline
\end{tabular}
\end{table} 
The result based on Theorem \ref{the} tends to give a sharper identification of important stocks. 
Note also that with Theorem \ref{the} significant stocks always lie within the same sector, while this is not the case with IND (Corollary \ref{coro:ind}) or with  \cite{wei2018}.
The above observations are particularly relevant for principal portfolios management mentioned in Section \ref{sec:intro} because they can help the investors to formulate more effectively their investment strategies on which stocks to long or short.
Specifically, it is interesting to point out that PC2 and PC3 based on Theorem \ref{the} are particularly relevant due to the clear contrast between the IT sector and the finance sector in PC2, and that between the IT sector and the health sector in PC3.   
Accordingly, 
the IT sector seems to be set well apart from other sectors, indicating different positions (long or short) from other sectors in the portfolio.
The different interpretations demonstrate the importance of an appropriate significance test of the loadings. 

We argue that without an appropriate test, users of PCA may be unaware of the subtle variabilities of loadings, and tend to  mis-interpret the eigenvectors. 

\section{Conclusion}\label{sec:con}
In this paper, by building on \cite{hannan1976asymptotic}, we have made explicit the asymptotic sampling properties of the eigenvalues and eigenvectors in PCA for stationary and ergodic multivariate linear time series and assessed their efficacy. 
The classical PCA results for independent data and the Taniguchi-Krishnaiah method are shown to be deficient for time series. Although they have led to simpler theoretical results, simulations have revealed their poor performance for time series data, especially for the eigenvectors. 
Thus, while the number of PCs is quite robust in respect of different dependence structure assumption, the interpretation of the principal component loadings may be questionable under these assumptions.
Direct and bootstrap method to estimate the asymptotic covariance have been given to draw inference in practice.
We argue that when applying a PCA to time series, the dependence of data should not be ignored and we have shown step by step how an appropriate inference can be conducted. 

Essentially, the shortcomings of the two defective methods for time series lie with the fact that they have ignored some information contained in cross covariances. Another way to exploit the omitted information is to adopt a frequency-domain PCA. See, e.g., \cite{brillinger1981time} and \cite{priestley1974}. However, unsolved problems remain with this approach, the possible non-causality of the principal components being a major one. They await further research.

\section*{Acknowledgements}\label{sec:ack}

We are most grateful to the Joint Editor and the referees for their constructive comments and suggestions.

\appendix

\section{Appendix}{\label{sec:app}}
In this appendix, we give more information of the simulation, including the details of DGPs and the results of the two additional simulation in Section \ref{subsec1}. 

\subsection{Details of the DGPs}\label{subsec:app1}
The coefficient matrices of the DGPs are shown as follows.

\textbf{DGP 1}: VAR(1).
\begin{align*}
\mathbf{F}(1)=&\left(
\begin{matrix}
 0.21 & -0.49 & 0.16 & -0.14 & -0.36 \\ 
 -0.25 & 0.074 & 0.38 & -0.14 & 0.047 \\ 
 -0.11 & 0.26 & 0.39 & 0.091 & 0.18 \\ 
 -0.41 & 0.37 & 0.066 & 0.37 & 0.028 \\ 
 0.46 & -0.46 & 0.094 & 0.18 & -0.41 \\ 
\end{matrix}
\right).
\end{align*}

\textbf{DGP 2, 5-8}: VMA(1).
\begin{align*}
\mathbf{G}(1)=&\left(
\begin{matrix}
    -0.46 & 0.17  & 0.23  & 0.40  & -0.22 \\
    0.15  & -0.59 & 0.05  & -0.26 & -0.24 \\
    0.13  & -0.32 & -0.26 & -0.28 & -0.41 \\
    0.15  & 0.20  & 0.51  & -0.38 & -0.55 \\
    0.43  & 0.02  & -0.25 & -0.32 & -0.34 \\
\end{matrix}
\right).
\end{align*}

\textbf{DGP 3}: VMA(2).
\begin{align*}
\mathbf{G}(1)=&\left(
\begin{matrix}
    0.62  & -0.73 & -0.02 & 0.52  & -0.52 \\
    0.71  & -0.11 & -0.19 & 0.43  & -0.75 \\
    0.09  & -0.56 & 0.79  & -0.51 & -0.08 \\
    0.14  & 0.90  & 0.07  & -0.53 & -0.27 \\
    -0.23 & -0.01 & 0.59  & 0.81  & 0.10 \\
\end{matrix}
\right),\\
\mathbf{G}(2)=&\left(
\begin{matrix}
    0.26  & -0.16 & -0.16 & -0.38 & 0.27 \\
    0.34  & -0.27 & -0.33 & -0.44 & 0.06 \\
    -0.22 & -0.18 & 0.62  & -0.08 & 0.19 \\
    0.26  & -0.30 & 0.15  & -0.24 & 0.02 \\
    0.11  & 0.21  & 0.21  & -0.14 & 0.45 \\
\end{matrix}
\right).
\label{eq:simu4}
\end{align*}

\textbf{DGP 4}: VMA(3).
\begin{align*}
\begin{split}
\mathbf{G}(1)=&\left(
\begin{matrix}
    0.94  & -0.35 & -0.49 & 0.17  & -0.18 \\
    0.58  & 0.35  & -0.43 & -0.29 & -0.36 \\
    0.42  & -0.16 & 1.07  & -0.28 & 0.39 \\
    0.59  & 0.41  & -0.38 & 0.27  & -0.03 \\
    0.18  & 0.66  & -0.28 & 0.42  & 0.91 \\
\end{matrix}
\right),\\
\mathbf{G}(2)=&\left(
\begin{matrix}
    -0.14 & -0.04 & -0.33 & 0.40  & 0.02 \\
    0.27  & -0.32 & -0.27 & 0.02  & -0.16 \\
    0.45  & -0.18 & 0.03  & -0.31 & 0.16 \\
    0.41  & 0.39  & -0.33 & -0.36 & 0.10 \\
    0.31  & 0.47  & -0.08 & 0.14  & -0.26 \\
\end{matrix}
\right),\\
\mathbf{G}(3)=&\left(
\begin{matrix}
    -0.16 & 0.24  & 0.10  & 0.22  & 0.17 \\
    -0.22 & 0.17  & 0.10  & 0.26  & 0.14 \\
    0.06  & -0.03 & -0.11 & -0.05 & -0.17 \\
    -0.11 & 0.02  & 0.01  & 0.19  & 0.11 \\
    0.14  & -0.11 & 0.16  & -0.21 & -0.24 \\
\end{matrix}
\right).
\end{split}
\end{align*}

As for the distributions of the noise, those of DGP 1-4 are shown in Section \ref{sec:simu}. Now we give details and parameters of the noise distribution of DGP 5-8. 

\textbf{DGP 5}: VMA(1) model with multivariate Gaussian noise, but 1\% of all the noises are outliers from another multivariate Gaussian distribution.

First we draw 2000 observations of $\mathbf{e}(t)$ from $\mathcal{N}({0}_5,10{\mathrm{I}}_5)$.
Then, 20 random numbers from 1 to 2000 are drawn, which will be the positions of the outliers. 
The original observations in the first 10 positions are replaced by random vectors from $\mathcal{N}({10}_5,10{\mathrm{I}}_5)$, and those in the last 10 positions are replaced by random vectors from $\mathcal{N}({-10}_5,10{\mathrm{I}}_5)$.

\textbf{DGP 6}: VMA(1) model with noise from a multivariate skew normal distribution.

The multivariate skew normal distribution is discussed by \cite{chan1986} and \cite{azzalini1999} and here we use the parameterization in the latter, $\mathcal{SN}_5({\xi, \Omega, \alpha})$ . We set ${\xi}={0}_5$, ${\Omega}=10{\mathrm{I}}_5$ and ${\alpha}=(1,2,3,4,5)$.

\textbf{DGP 7}: VMA(1) model with noise from a multivariate $t(5)$ distribution.

\textbf{DGP 8}: VMA(1) model with noise from a multivariate $t(8)$ distribution.

For the multivariate \textit{t}-distribution $t_v({\mu, \Sigma})$, we set ${\mu}={0}_5$, ${\Sigma}=10{\mathrm{I}}_5$, $v=5$ for DGP 7 and $v=8$ for DGP 8.

\subsection{Additional simulation}\label{subsec:app2}

In this subsection, we present two additional simulations, namely DGP 3 and 4, to enrich Subsection \ref{subsec1} and assess the efficacy of the asymptotics for Gaussian processes.

For each model, the vectorized $\Cov\left(l_{k}, l_{k'}\right)$ and $\Delta^*({a}_{kk'})$ in (\ref{eq:deltaa}) for $k,k'=1,...,5$ are each a $25$-dimensional vector. 
See their plots in Fig.\ref{fig23} for DGP 3 and 4 respectively.
\begin{figure}[!htbp]
\centering
\makebox{\includegraphics[width=0.8\textwidth]{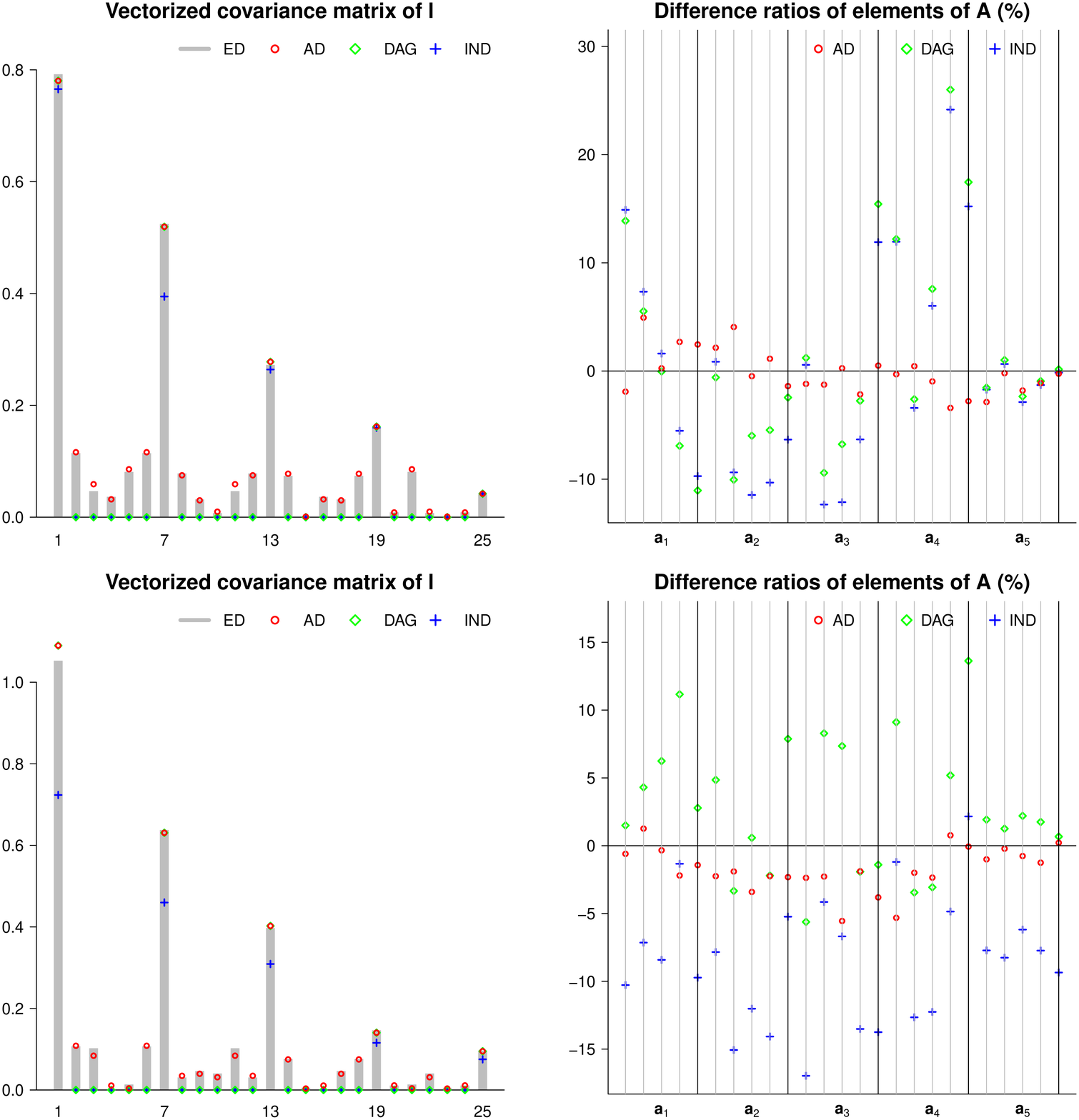}}
\caption{The vectorized $\Cov\left(l_{k}, l_{k'}\right)$ and $\Delta({a}_{kk'})$ for $k,k'=1,...,5$ for DGP 3 (Top pane) and DGP 4 (Bottom pane).} 
\label{fig23}
\end{figure}
For the vectorized $\Cov\left(l_{k}, l_{k'}\right)$, we can see AD fits ED well. IND tends to give a smaller $\Var\left(l_{k}\right)$, while DAG gives the same $\Var\left(l_{k}\right)$ as AD, both cases being consistent with the theoretical results.
However, for $\Cov\left(l_{k}, l_{k'}\right), k \neq k'$, IND and DAG claim that $l_{k}$ are uncorrelated, which is not supported by ED.

For $\sigma({a}_{kk'})$, AD performs pretty well and the difference ratios are controlled within 6\%. In contrast, IND and DAG have poor performance; the difference ratio can even reach a figure as big as 20\% for DGP 3.

As for the proportion of variation, Table \ref{table:prop23} shows $\Delta({r}_k)$ subscripted by AD, DAG and IND.
\begin{table}
\caption{\label{table:prop23}The difference $\Delta({r}_k)$ in percentage for $k=1,...,4$ for DGP 3 and DGP 4.}
\centering
\footnotesize
\begin{tabular}{rrrrr|rrrr}
  \hline
   & \multicolumn{4}{c|}{DGP1} & \multicolumn{4}{c}{DGP2}\\
 $k$ & 1 & 2 & 3 & 4 & 1 & 2 & 3 & 4 \\
  \hline
AD & -0.007 & -0.000 & 0.000 & 0.000 & 0.011 & 0.010 & 0.003 & 0.003 \\  
DAG & 0.044 & 0.026 & 0.043 & 0.020 & 0.035 & 0.011 & 0.018 & 0.005 \\ 
IND & 0.031 & 0.008 & 0.035 & 0.018 & -0.068 & -0.069 & -0.028 & -0.023 \\ \hline
\end{tabular}
\end{table}
Among the three asymptotic standard deviations, $\sigma_{\mathrm{AD}}({r}_k)$ is most close to $\sigma_{\mathrm{ED}}({r}_k)$, with the differences controlled within 0.02\%.
Although the differences for $\sigma_{\mathrm{DAG}}({r}_k)$ and $\sigma_{\mathrm{IND}}({r}_k)$ are larger, they are all controlled within 0.1\%, which is negligible.

\section*{Supplementary material}
\label{SM}
Supplementary material includes R codes for the empirical example in Section \ref{sec:eg}.

\renewcommand\bibname{\large \bf References}

\end{document}